\documentclass[runningheads]{llncs}
\usepackage[T1]{fontenc}
\usepackage{graphicx}
\usepackage{xcolor}

\usepackage{amsmath}
\usepackage{amssymb}

\usepackage{shuffle}

\usepackage{enumitem}
\usepackage{etoolbox}
\usepackage[bookmarks, hypertexnames = false, pdfencoding = auto, psdextra]{hyperref}
\usepackage[capitalize, nameinlink, noabbrev]{cleveref}

\usepackage{tikz}
\usetikzlibrary{fit}

\usepackage{thm-restate}

\newcommand{\set}[1]{\{ #1 \}}               % set
\newcommand{\setd}[2]{\{ #1 \mid #2 \}}      % set (with divider)
\newcommand{\tuple}[1]{\langle #1 \rangle}   % tuple
\newcommand{\union}{\cup}

\newcommand{\emptyword}{\varepsilon}
\newcommand{\emptylang}{\emptyset}

\newcommand{\myrule}[3][]{\ifstrempty{#1}{\dfrac{#2}{#3}}{\dfrac{#2}{#3}[#1]}} % rule, possibly named

\newcommand{\concurrent}{\not \gtrless}
\newcommand{\lang}[1]{L(#1)}

\newcommand{\shuff}[2]{\shuffle_{#1}(#2)}
\newcommand{\shuffx}[3]{\shuffle^{#1}_{#2}(#3)}

\newcommand{\causes}{\leq}

\tikzstyle{event} = [
  inner ysep = 0,
  inner xsep = 2pt
]

\newcommand{\pomsetbox}[2][dummy]{%
  \node[draw, inner ysep = 5pt, inner xsep = 3pt, rounded corners,% fill=white,
        fit = #2] (#1) {};%
}

\usepackage{lineno}
\def\techreport{}

\begin{document}
\ifdefined\techreport
\title{Shuffling posets on trajectories \\ (technical report)}
\else
\title{Shuffling posets on trajectories}
\fi
%
%\titlerunning{Abbreviated paper title}
% If the paper title is too long for the running head, you can set
% an abbreviated paper title here
%
\author{Luc Edixhoven\orcidID{0000-0002-6011-9535}}
%
% \authorrunning{F. Author et al.}
% First names are abbreviated in the running head.
% If there are more than two authors, 'et al.' is used.
%
\institute{
Open University of the Netherlands
\and
Centrum Wiskunde \& Informatica (CWI)
\\
\email{luc.edixhoven@ou.nl}
% \url{http://www.springer.com/gp/computer-science/lncs} \and
% ABC Institute, Rupert-Karls-University Heidelberg, Heidelberg, Germany\\
% \email{\{abc,lncs\}@uni-heidelberg.de}
}
\maketitle              % typeset the header of the contribution
\begin{abstract}
Choreographies describe possible sequences of interactions among a set of agents.
We aim to join two lines of research on choreographies: the use of the shuffle on trajectories operator to design more expressive choreographic languages, and the use of models featuring partial orders, to compactly represent concurrency between agents.
Specifically, in this paper, we explore the application of the shuffle on trajectories operator to individual posets, and we give a characterisation of shuffles of posets which again yield an individual poset.
% The design and analysis of communication protocols in distributed systems is a much researched and complex topic.
% The shuffle on trajectories operator can be used to define more expressive specification languages for these communication protocols, specifically choreographic languages.
% Branching pomsets were recently introduced as a new, pomset-based model to compactly represent and then analyse the behaviour of communication protocols.
% It would be a valuable next step to give an encoding of the shuffle on trajectories operator in terms of branching pomsets.
% In this paper, as a first step, we discuss the shuffle on trajectories operator in terms of traditional (non-branching) pomsets.
% Specifically, we give a characterisation of pomsets which, when shuffled, will yield a single pomset, rather than a set of pomsets.

\keywords{Posets \and Shuffle on trajectories \and Concurrency}
\end{abstract}

\bigskip

\ifdefined\techreport
\noindent Technical report of a paper to be published in the proceedings of iFM 2023.
\else
\fi

\section{Introduction}
\label{sec:intro}

Distributed systems are becoming ever more important. However, designing and implementing them is difficult. The complexity resulting from concurrency and dependencies among agents makes the process error-prone and debugging non-trivial. As a consequence, much research has been dedicated to analysing communication patterns, or protocols, among sets of agents in distributed systems.
Examples of such research goals are to show the presence or absence of certain safety properties in a given system, to automate such analysis, and to guarantee the presence of desirable properties by construction.

Part of this research deals with \emph{choreographies}. Choreographies can be used as global specifications for asynchronously communicating agents, and contain certain safety properties by construction.
As a drawback, choreographic languages typically have limitations on their expressiveness, since they rely on grammatical constructs for their safety properties, which exclude some communication patterns.
We have recently shown that the \emph{shuffle on trajectories} operator can be used to specify choreographies without compromising expressiveness~\cite{DBLP:journals/ijfcs/EdixhovenJ23}. Consequently, it could serve as a basis for more expressive choreographic languages.

Other recent work on choreographies includes the use of models featuring partial orders, such as event structures~\cite{DBLP:journals/jlap/CastellaniDG23} and pomsets~\cite{DBLP:journals/jlp/TuostoG18,DBLP:journals/corr/abs-2208-04632}, to represent and analyse the behaviour of choreographies.
By using a partial order to explicitly capture causal dependencies between pairs of actions, these models avoid the exponential blowup from, e.g., parallel composition of finite state machines.

We aim to join these two lines of research by extending the shuffle on trajectories operator from words, i.e., totally ordered traces, and languages to partially ordered traces and sets thereof.
In this paper, as a first step, we explore the application of the shuffle on trajectories operator to individual partially ordered sets, or posets. The main challenge is that the resulting behaviour cannot always be represented as one poset and may require a set of them. In particular, we give a characterisation of shuffles of posets which again yield an individual poset.

\paragraph{Outline}
We recall the concept and definition of the shuffle on trajectories operator in \cref{sec:shuffle}.
We briefly discuss posets in \cref{sec:posets}.
In \cref{sec:shuffle-posets} we discuss how to apply the shuffle on trajectories operator to posets, and specifically which shuffles of posets will yield an individual poset as a result.
Finally, we briefly discuss future work in \cref{sec:future}.

The proofs of \cref{prop:extract} and \cref{lem:relation} can be found in
\ifdefined\techreport
the appendix.
\else
a separate technical report on arXiv (\url{https://arxiv.org/}).
\fi

\section{Shuffle on trajectories}
\label{sec:shuffle}

We recall the basic definitions from~\cite{DBLP:journals/ijfcs/EdixhovenJ23}.
The shuffle on trajectories operator is a powerful variation of the traditional shuffle operator\footnote{In concurrency theory, the shuffle operator is also known as free interleaving, non-communication merge, or parallel composition.}, which adds a control trajectory (or a set thereof) to restrict the permitted orders of interleaving. This allows for fine-grained control over orderings when shuffling words or languages.
The binary operator was defined --- and its properties thoroughly studied --- by Mateescu et al.~\cite{DBLP:journals/tcs/MateescuRS98}; a multiary variant was introduced slightly later~\cite{DBLP:journals/jalc/MateescuSY00}.

When defined on words, the shuffle on trajectories takes $n$ words and a \emph{trajectory}, which is a word over the alphabet $\set{1, \ldots, n}$. This trajectory specifies the exact order of interleaving of the shuffled words: in \cref{fig:bananapear}, the trajectory 1221112112 specifies that the result should first take a symbol from the first word, then from the second, then again from the second and so on.

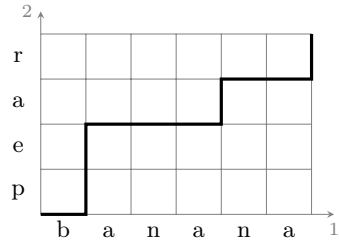
\begin{figure}[t]
	\centering
	\begin{tikzpicture}[scale=.6, >=stealth]
		% Grid
		\draw[help lines] (0,0) grid (6,4);
		\draw[help lines, ->] (6,0) -- (6.5,0) node[below] {\scriptsize 1};
		\draw[help lines, ->] (0,4) -- (0,4.5) node[left] {\scriptsize 2};

		% Trajectory
		\draw[very thick] (0,0) -- (1,0) -- (1,2) -- (4,2) -- (4,3) -- (6,3) -- (6,4);

		% Labels
		\node[anchor=base] at (0.5,-.5) {b};
		\node[anchor=base] at (1.5,-.5) {a};
		\node[anchor=base] at (2.5,-.5) {n};
		\node[anchor=base] at (3.5,-.5) {a};
		\node[anchor=base] at (4.5,-.5) {n};
		\node[anchor=base] at (5.5,-.5) {a};
		\node at (-.5,0.5) {p};
		\node at (-.5,1.5) {e};
		\node at (-.5,2.5) {a};
		\node at (-.5,3.5) {r};

		% Expression
		% \node[anchor=base] at (-5,2.5) {$\shuffxxx{2}{1221112112}{\textup{abcdef}, \textup{ghij}}$};
		% \node[anchor=base] at (-5,1.5) {$= \textup{aghbcdiefj}$};
	\end{tikzpicture}

	\caption{The shuffle of `banana' and `pear' over a trajectory $1221112112$: `bpeanaanar'.}
	\label{fig:bananapear}
\end{figure}

Formally, let $w_1, \ldots, w_n$ be finite words over some alphabet and let $t$ be a finite word over the alphabet $\set{1, \ldots, n}$. Let $\emptyword$ be the empty word. Then:
\[
	\shuffx{n}{t}{w_1, \ldots, w_n} =
	\begin{cases}
		a \shuffx{n}{t'}{w_1, \ldots, w_i', \ldots, w_n} & \text{if $t = it'$ and $w_i = a w_i'$}
		\\
		\emptyword & \text{if $t = w_1 = \ldots = w_n = \emptyword$}
	\end{cases}
\]
We note that $\shuffx{n}{t}{w_1, \ldots, w_n}$ is only defined if the number of occurrences of $i$ in $t$ precisely matches the length of $w_i$ for every $i$. We then say that $t$ \emph{fits} $w_i$.

\begin{example}
	\strut
	\begin{itemize}
		\item $\shuffx{3}{121332}{ab, cd, \mathit{ef}} = \mathit{acbefd}$, since 121332 fits every word.
		\item $\shuffx{2}{121}{ab, cd}$ is undefined, since 121 does not fit $cd$.
	\end{itemize}
\end{example}

The shuffle on trajectories operator naturally generalises to languages: the shuffle of a number of languages on a set (i.e., a language) of trajectories is defined as the set of all valid shuffles of words in the languages for which the trajectory fits all the words.
Formally:
\[
	\shuffx{n}{T}{L_1, \ldots, L_n} = \setd{\shuffx{n}{t}{w_1, \ldots, w_n}}{t \in T, w_1 \in L_1, \ldots, w_n \in L_n}
\]

As the operator's arity is clear from its operands, we typically omit it.

\section{Posets}
\label{sec:posets}

Partially ordered sets, or posets for short, consist of a set of nodes $E$ (events), and a partial order\footnote{Recall that a partial order is reflexive, transitive and antisymmetric.} $\causes$ defining dependencies between pairs of events --- i.e., an event can only fire if all events preceding it in the partial order have already fired.
We write $a < b$ to denote that $a \leq b$ and $a \neq b$. We write $a \geq b$ resp. $a > b$ to denote that $b \leq a$ resp. $b < a$. We write $a \concurrent b$ to denote that $a \not \leq b$ and $b \not \leq a$; we then say that $a$ and $b$ are \emph{concurrent}.
We occasionally write $E_P, \causes_P$, $<_P$, $\geq_P$, $>_P$ and $\concurrent_P$ to specify that the set of events or relation belongs to poset $P$, but where this is clear from context we typically omit the subscript.

The behaviour (or language) of a poset $P$, written $\lang{P}$, is the set of all maximal traces, i.e., maximal sequences of its events, that abide by $\causes$.
In this sense, posets can be considered a generalisation of words with concurrency: they feature a fixed set of symbols (events)\footnote{There is a difference between words and posets in the sense that the events in a poset must be unique, whereas a word may contain duplicate symbols. It would be more accurate to say that words generalise to \emph{labelled} posets, or lposets, and from there to partially ordered \emph{multisets}, or pomsets. However, in this paper we focus on posets, where all symbols are thus unique.}, but they can allow multiple orderings of them instead of only a single one. Concurrent events can happen in any order. Consequently, all traces obtained from a trace in $\lang{P}$ by swapping adjacent concurrent events must also be in $\lang{P}$. In fact, any trace in $\lang{P}$ can be obtained from any other trace in $\lang{P}$ in this fashion.

\begin{example}
\label{ex:poset}
	For poset $P_{ex}$ in \cref{fig:posets}, $E = \set{a, b, c, d}$ and the partial order consists of $a \causes a$, $a \causes c$, $a \causes d$, $b \causes b$, $b \causes d$, $c \causes c$ and $d \causes d$.
	Its language $\lang{P_{ex}}$ consists of the traces $abcd$, $abdc$, $acbd$, $bacd$, and $badc$.
\end{example}

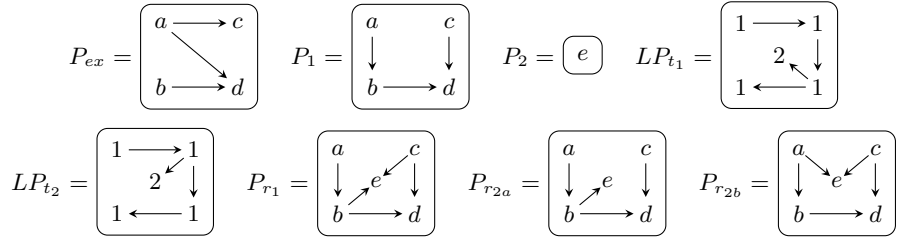
\begin{figure}[t]
	\centering

	\def\h{1.2}
	\def\v{.5}

	$P_{ex} = $
	\begin{tikzpicture}[->, >=stealth, scale = 0.85, baseline = (base.south)]
		\node (base) at (0, 0) {};

		\node[event] (a) at (0, \v) {$a$};
		\node[event] (b) at (0, -\v) {$b$};
		\node[event] (c) at (\h, \v) {$c$};
		\node[event] (d) at (\h, -\v) {$d$};

		\draw[shorten < = 2pt] (a) -- (d);
		\draw (a) -- (c);
		\draw (b) -- (d);

		\pomsetbox{(a) (d)}
	\end{tikzpicture}
	\quad
	$P_1 = $
	\begin{tikzpicture}[->, >=stealth, scale = 0.85, baseline = (base.south)]
		\node (base) at (0, 0) {};

		\node[event] (a) at (0, \v) {$a$};
		\node[event] (b) at (0, -\v) {$b$};
		\node[event] (c) at (\h, \v) {$c$};
		\node[event] (d) at (\h, -\v) {$d$};

		\draw[shorten < = 3pt, shorten > = 3pt] (a) -- (b);
		\draw (b) -- (d);
		\draw[shorten < = 3pt, shorten > = 3pt] (c) -- (d);
		% \draw (c) -- (e);

		\pomsetbox{(a) (d)}
	\end{tikzpicture}
	\quad
	$P_2 = $
	\begin{tikzpicture}[->, >=stealth, scale = 0.85, baseline = (base.south)]
		\node (base) at (0, 0) {};

		\node[event] (e) at (0, 0) {$e$};

		\pomsetbox{(e)}
	\end{tikzpicture}
	\quad
	$LP_{t_1} = $
	\begin{tikzpicture}[->, >=stealth, scale = 0.85, baseline = (base.south)]
		\node (base) at (0, 0) {};

		\node[event] (a) at (0, \v) {$1$};
		\node[event] (b) at (0, -\v) {$1$};
		\node[event] (c) at (\h, \v) {$1$};
		\node[event] (d) at (\h, -\v) {$1$};
		\node[event] (e) at (.5*\h, 0) {$2$};

		\draw (a) -- (c);
		\draw[shorten < = 3pt, shorten > = 3pt] (c) -- (d);
		\draw (d) -- (b);
		\draw (d) -- (e);

		\pomsetbox{(a) (d)}
	\end{tikzpicture}

	\medskip

	$LP_{t_2} = $
	\begin{tikzpicture}[->, >=stealth, scale = 0.85, baseline = (base.south)]
		\node (base) at (0, 0) {};

		\node[event] (a) at (0, \v) {$1$};
		\node[event] (b) at (0, -\v) {$1$};
		\node[event] (c) at (\h, \v) {$1$};
		\node[event] (d) at (\h, -\v) {$1$};
		\node[event] (e) at (.5*\h, 0) {$2$};

		\draw (a) -- (c);
		\draw[shorten < = 3pt, shorten > = 3pt] (c) -- (d);
		\draw (d) -- (b);
		\draw (c) -- (e);

		\pomsetbox{(a) (d)}
	\end{tikzpicture}
	\quad
	$P_{r_1} = $
	\begin{tikzpicture}[->, >=stealth, scale = 0.85, baseline = (base.south)]
		\node (base) at (0, 0) {};

		\node[event] (a) at (0, \v) {$a$};
		\node[event] (b) at (0, -\v) {$b$};
		\node[event] (c) at (\h, \v) {$c$};
		\node[event] (d) at (\h, -\v) {$d$};
		\node[event] (e) at (0.5*\h, 0) {$e$};

		\draw[shorten < = 3pt, shorten > = 3pt] (a) -- (b);
		\draw (b) -- (d);
		\draw[shorten < = 3pt, shorten > = 3pt] (c) -- (d);
		\draw (b) -- (e);
		\draw (c) -- (e);

		\pomsetbox{(a) (d)}
	\end{tikzpicture}
	\quad
	$P_{r_{2a}} = $
	\begin{tikzpicture}[->, >=stealth, scale = 0.85, baseline = (base.south)]
		\node (base) at (0, 0) {};

		\node[event] (a) at (0, \v) {$a$};
		\node[event] (b) at (0, -\v) {$b$};
		\node[event] (c) at (\h, \v) {$c$};
		\node[event] (d) at (\h, -\v) {$d$};
		\node[event] (e) at (0.5*\h, 0) {$e$};

		\draw[shorten < = 3pt, shorten > = 3pt] (a) -- (b);
		\draw (b) -- (d);
		\draw[shorten < = 3pt, shorten > = 3pt] (c) -- (d);
		\draw (b) -- (e);

		\pomsetbox{(a) (d)}
	\end{tikzpicture}
	\quad
	$P_{r_{2b}} = $
	\begin{tikzpicture}[->, >=stealth, scale = 0.85, baseline = (base.south)]
		\node (base) at (0, 0) {};

		\node[event] (a) at (0, \v) {$a$};
		\node[event] (b) at (0, -\v) {$b$};
		\node[event] (c) at (\h, \v) {$c$};
		\node[event] (d) at (\h, -\v) {$d$};
		\node[event] (e) at (0.5*\h, 0) {$e$};

		\draw[shorten < = 3pt, shorten > = 3pt] (a) -- (b);
		\draw (b) -- (d);
		\draw[shorten < = 3pt, shorten > = 3pt] (c) -- (d);
		\draw (a) -- (e);
		\draw (c) -- (e);

		\pomsetbox{(a) (d)}
	\end{tikzpicture}

	\caption{Graphical representation of a number of posets and lposets, where an arrow from $a$ to $b$ should be read as $a \causes b$. The partial order is the reflexive and transitive closure of the dependencies depicted by the arrows. For the lposets, the labels are shown rather than their events.}
	\label{fig:posets}
\end{figure}

We note that the dependencies in a poset can also be observed in its set of traces. For example, if $a < b$ then $a$ will precede $b$ in every trace, and if $a \concurrent b$ then there will both be traces where $a$ precedes $b$ and traces where $b$ precedes $a$.
Formally, we can extract the following relation $\causes_L$ from a set of traces $L \subseteq E^*$:
\begin{gather*}
	\myrule{
		\begin{gathered}
			\exists x, y, z \in E^*: xaybz \in L
			\\
			\forall \hat{x}, \hat{y}, \hat{z} \in E^*: \hat{x}b\hat{y}a\hat{z} \notin L
		\end{gathered}
	}{
		a \causes_L b
	}
	\qquad
	\myrule{
		% empty
	}{
		a \causes_L a
	}
	\qquad
	\myrule{
		a \causes_L b \causes_L c
	}{
		a \causes_L c
	}
\end{gather*}

We then propose the following:

% \begin{proposition}
\begin{restatable}{proposition}{propextract}
\label{prop:extract}
	Let $P = \tuple{E_P, \causes_P}$ be a poset.
	Then ${\causes_{\lang{P}}} = {\causes_P}$.
% \end{proposition}
\end{restatable}

To model trajectories, which require duplicate symbols, we must also introduce \emph{labelled} posets, or \emph{lposets}. In these, every event is assigned a label, which is not necessarily unique. Its traces then use these labels instead of the events.

\section{Shuffling posets}
\label{sec:shuffle-posets}

As a first step towards shuffling posets, we first reinterpret shuffles on words as posets. In other words: we consider the case where all posets, including the trajectory, are totally ordered and thus consist of a single trace.
This is shown in \cref{fig:bananapear-poset}, which features the shuffle from \cref{fig:bananapear} interpreted as a poset. The traces `banana' and `pear' are present as totally ordered parts of the poset, and the trajectory adds additional dependencies between the two, as shown by the vertical (and diagonal) arrows.

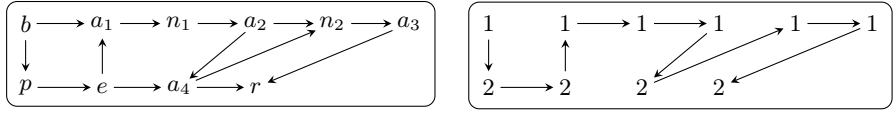
\begin{figure}[t]
	\centering

	\def\h{1.2}
	\def\v{.5}

	\begin{tikzpicture}[->, >=stealth, scale = 0.85, baseline = (base.south)]
		\node (base) at (0, 0) {};

		\node[event] (b) at (0, \v) {$b$};
		\node[event] (a1) at (\h, \v) {$a_1$};
		\node[event] (n1) at (2*\h, \v) {$n_1$};
		\node[event] (a2) at (3*\h, \v) {$a_2$};
		\node[event] (n2) at (4*\h, \v) {$n_2$};
		\node[event] (a3) at (5*\h, \v) {$a_3$};
		\node[event] (p) at (0, -\v) {$p$};
		\node[event] (e) at (\h, -\v) {$e$};
		\node[event] (a4) at (2*\h, -\v) {$a_4$};
		\node[event] (r) at (3*\h, -\v) {$r$};

		\draw (b) -- (a1);
		\draw (a1) -- (n1);
		\draw (n1) -- (a2);
		\draw (a2) -- (n2);
		\draw (n2) -- (a3);
		\draw (p) -- (e);
		\draw (e) -- (a4);
		\draw (a4) -- (r);

		\draw[shorten < = 3pt, shorten > = 3pt] (b) -- (p);
		\draw[shorten < = 3pt, shorten > = 3pt] (e) -- (a1);
		\draw[shorten < = 1pt, shorten > = 1pt] (a2) -- (a4);
		\draw (a4) -- (n2);
		\draw (a3) -- (r);

		% \draw[shorten < = 3pt, shorten > = 3pt] (a) -- (b);
		% \draw (b) -- (d);
		% \draw[shorten < = 3pt, shorten > = 3pt] (c) -- (d);

		\pomsetbox{(b) (r) (a3)}
	\end{tikzpicture}
	\quad
	\begin{tikzpicture}[->, >=stealth, scale = 0.85, baseline = (base.south)]
		\node (base) at (0, 0) {};

		\node[event] (b) at (0, \v) {1};
		\node[event] (a1) at (\h, \v) {1};
		\node[event] (n1) at (2*\h, \v) {1};
		\node[event] (a2) at (3*\h, \v) {1};
		\node[event] (n2) at (4*\h, \v) {1};
		\node[event] (a3) at (5*\h, \v) {1};
		\node[event] (p) at (0, -\v) {2};
		\node[event] (e) at (\h, -\v) {2};
		\node[event] (a4) at (2*\h, -\v) {2};
		\node[event] (r) at (3*\h, -\v) {2};

		% \node[event] (b) at (0, \v) {$b$};
		% \node[event] (a1) at (\h, \v) {$a_1$};
		% \node[event] (n1) at (2*\h, \v) {$n_1$};
		% \node[event] (a2) at (3*\h, \v) {$a_2$};
		% \node[event] (n2) at (4*\h, \v) {$n_2$};
		% \node[event] (a3) at (5*\h, \v) {$a_3$};
		% \node[event] (p) at (0, -\v) {$p$};
		% \node[event] (e) at (\h, -\v) {$e$};
		% \node[event] (a4) at (2*\h, -\v) {$a_4$};
		% \node[event] (r) at (3*\h, -\v) {$r$};

		% \draw (b) -- (a1);
		\draw (a1) -- (n1);
		\draw (n1) -- (a2);
		% \draw (a2) -- (n2);
		\draw (n2) -- (a3);
		\draw (p) -- (e);
		% \draw (e) -- (a4);
		% \draw (a4) -- (r);

		\draw[shorten < = 3pt, shorten > = 3pt] (b) -- (p);
		\draw[shorten < = 3pt, shorten > = 3pt] (e) -- (a1);
		\draw[shorten < = 1pt, shorten > = 1pt] (a2) -- (a4);
		\draw (a4) -- (n2);
		\draw (a3) -- (r);

		% \draw[shorten < = 3pt, shorten > = 3pt] (a) -- (b);
		% \draw (b) -- (d);
		% \draw[shorten < = 3pt, shorten > = 3pt] (c) -- (d);

		\pomsetbox{(b) (r) (a3)}
	\end{tikzpicture}

	\caption{The figure on the left shows the shuffle from \cref{fig:bananapear} interpreted as a shuffle of posets. Indices have been added to duplicate symbols to make them unique. Some of the arrows are redundant but are kept to illustrate the general idea. The figure on the right shows the trajectory, 1221112112, as an lposet.}
	\label{fig:bananapear-poset}
\end{figure}

Generalising this to arbitrary posets and lposets is not trivial, but we have some knowledge to assist us. Crucially, since we can determine the language of a poset, it must be so that the result of shuffling posets yields the same language as the shuffle of the languages of these posets, which is defined in \cref{sec:shuffle}:
\[
	\lang{\shuff{P_t}{P_1, \ldots, P_n}} = \shuff{\lang{P_t}}{\lang{P_1}, \ldots, \lang{P_n}}
\]
If the result is an individual poset, by \cref{prop:extract} it must then be:
\[
	\shuff{P_t}{P_1, \ldots, P_n} = \tuple{E_{P_1} \union \ldots \union E_{P_n}, \causes_{\shuff{\lang{P_t}}{\lang{P_1}, \ldots, \lang{P_n}}}}
\]

For example, consider $\shuff{LP_{t_1}}{P_1, P_2}$, with $LP_{t_1}$, $P_1$ and $P_2$ as in \cref{fig:posets}. $LP_{t_1}$ has traces 1121 and 1112, $P_1$ has traces $abcd$, $acbd$ and $cabd$, and $P_2$ has a single trace $e$. Shuffling these languages yields $L_1 = \set{abced, acbed, cabed, abcde,\allowbreak acbde,\allowbreak cabde}$. From this we extract $\causes_{L_1}$, which contains all the dependencies present in $P_1$ and $P_2$ and, additionally, $a \causes_{L_1} e$, $b \causes_{L_1} e$ and $c \causes_{L_1} e$. This corresponds to poset $P_{r_1}$ in \cref{fig:posets}, which indeed yields the language $L_1$.

However, now consider $\shuff{LP_{t_2}}{P_1, P_2}$, again as in \cref{fig:posets}. $LP_{t_2}$ has traces 1211, 1121 and 1112, which yields $L_2 = L_1 \union \set{abecd, acebd, caebd}$. From this we extract $\causes_{L_2}$, which still contains all the dependencies in $P_1$ and $P_2$, but otherwise only $a \causes_{L_2} e$: the traces $abecd$ and $acebd$ imply that $b$ and $c$ are concurrent with $e$. However, then the trace $aebcd$ should also be in $L_2$, which it is not. We can then conclude from \cref{prop:extract} that there exists no poset $P$ such that $\lang{P} = L_2$. In fact, $L_2$ corresponds to a set of two posets, namely $P_{r_{2a}}$ and $P_{r_{2b}}$ in \cref{fig:posets}.

We proceed by giving a characterisation of shuffles of posets for which the result corresponds to an individual poset. A key insight is that, if the result must correspond to an individual poset, then any two events which are concurrent in one of the operands of the shuffle must, in the resulting poset, have the same relation ($<$, $>$ or $\concurrent$) to any third event originating from another operand:
% Formally, we prove the following lemma, whose proof can be found in \cref{sec:proofs}:

\begin{restatable}{lemma}{lemrelation}
\label{lem:relation}
	Let $LP_t$ be an lposet and $P_1, \ldots, P_n, P$ posets such that $\lang{\shuff{LP_t}{P_1,\allowbreak \ldots, P_n}} = \lang{P}$ and $\lang{P} \neq \emptylang$.
	If $a, b \in E_{P_i}$ such that $a \concurrent_{P_i} b$ and $c \in E_{P_j}$ with $i \neq j$, then either $a, b <_P c$, or $a, b >_P c$, or $a, b \concurrent_P c$.
\end{restatable}
% \begin{proof}[sketch]
% 	We show that all other pairs of relations between $a$, $b$ and $c$ lead to a contradiction. We first observe that, because of how the shuffle on trajectories operator works, (1) the ordering (or lack thereof) between two events in $P_i$ must be the same in $P$, and (2) if two events are concurrent in $P_i$ and can thus be swapped when adjacent in traces of $P_i$, then they can be swapped in traces of $P$ when there are no other events from $P_i$ in between them.

% 	Orderings $a <_P c <_P b$ and $b <_P c <_P a$ both contradict (1), which implies that $a \concurrent_P b$.
% 	Otherwise, if one of $a$ and $b$ is concurrent with $c$, then we use (2) to construct traces of $P$ to show that, since $a \concurrent_{P_i} b$, both $a$ and $b$ can occur on both sides of $c$, which implies that both of them must be concurrent with $c$.
% 	\qed
% \end{proof}

We can then group the events in every $P_i$ according to the reflexive and transitive closure of the concurrency relation $\concurrent_{P_i}$; two events which are related in this closure then belong to the same group. Note that, while the events in a group are partially ordered, the groups of every $P_i$ are, by construction, totally ordered. It follows from \cref{lem:relation} that two events in the same group, even when not concurrent, must have the same relation to any event outside of their group in $P$. This in turn implies a similar condition on the trajectory lposet: any two $i$-labelled events in $LP_t$ that can match two events from the same group of $P_i$ must have the same relation to any $j$-labelled event in $LP_t$ (where $j$ is not necessarily unequal to $i$) that can match an event outside of their group.

\cref{fig:shuffle-poset} shows $P_{r_1}$, corresponding to $\shuff{LP_{t_1}}{P_1, P_2}$, and $LP_{t_1}$ (from \cref{fig:posets}), both restructured to show the groups of $P_1$ and $P_2$. This demonstrates an interesting parallel with \cref{fig:bananapear-poset}: both feature horizontal traces with additional arrows specifying dependencies between components of these traces. However, in \cref{fig:bananapear-poset} the components consist of individual events, whereas in \cref{fig:shuffle-poset} the components consist of posets.
In this sense, shuffles resulting in individual posets generalise shuffles on traces.

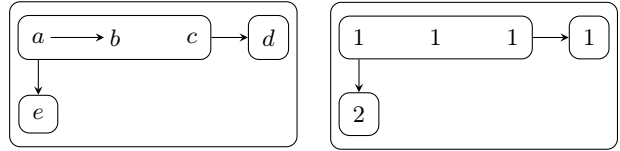
\begin{figure}[t]
	\centering

	\def\h{1.2}
	\def\v{.6}

	\begin{tikzpicture}[->, >=stealth, scale = 0.85, baseline = (base.south)]
		\node (base) at (0, 0) {};

		\node[event] (a) at (0, \v) {$a$};
		\node[event] (b) at (\h, \v) {$b$};
		\node[event] (c) at (2*\h, \v) {$c$};
		\node[event] (d) at (3*\h, \v) {$d$};
		\node[event] (e) at (0, -\v) {$e$};

		\draw (a) -- (b);

		\pomsetbox[abc]{(a) (b) (c)}
		\pomsetbox[dd]{(d)}
		\pomsetbox[ee]{(e)}

		\draw (abc) -- (dd);
		\draw (abc.south -| a) -- (ee);

		\pomsetbox{(abc) (ee) (dd)}
	\end{tikzpicture}
	\quad
	\begin{tikzpicture}[->, >=stealth, scale = 0.85, baseline = (base.south)]
		\node (base) at (0, 0) {};

		\node[event] (a) at (0, \v) {1};
		\node[event] (b) at (\h, \v) {1};
		\node[event] (c) at (2*\h, \v) {1};
		\node[event] (d) at (3*\h, \v) {1};
		\node[event] (e) at (0, -\v) {2};

		\pomsetbox[abc]{(a) (b) (c)}
		\pomsetbox[dd]{(d)}
		\pomsetbox[ee]{(e)}

		\draw (abc) -- (dd);
		\draw (abc.south -| a) -- (ee);

		\pomsetbox{(abc) (ee) (dd)}
	\end{tikzpicture}

	\caption{The figure on the left shows $P_{r_1}$, corresponding to $\shuff{LP_{t_1}}{P_1, P_2}$ (see \cref{fig:posets}), restructured to show the groups of $P_1$ and $P_2$. An arrow from one group of events to another should be read as an arrow from all events in the originating group to all events in the target group. The figure on the right shows the restructured $LP_{t_1}$; the dependencies within groups in $LP_{t_1}$ are irrelevant for the resulting traces.}
	\label{fig:shuffle-poset}
\end{figure}

Concluding, we can then characterise shuffles on posets which result in individual posets as those where the trajectory lposet is structured along the operand posets' groups, as in \cref{fig:shuffle-poset}, possibly with dependencies between different operands' groups.

\section{Future work}
\label{sec:future}

Now that we have studied shuffles of posets resulting in individual posets, there are two evident avenues for future work: (1) shuffles of lposets, where one label may occur multiple times rather than just considering orderings of unique events and (2) shuffles of posets resulting in sets of posets and shuffles of sets of posets, where the main challenge may be to minimise the resulting number of posets.

%
% ---- Bibliography ----
%
% BibTeX users should specify bibliography style 'splncs04'.
% References will then be sorted and formatted in the correct style.
%
\bibliographystyle{splncs04}
\bibliography{ifmphd2023}

\ifdefined\techreport

\clearpage
\appendix

\section{Proofs from the paper}
\label{sec:proofs}

\propextract*
\begin{proof}
	\strut
	\begin{itemize}
		\item Suppose that $a \causes_P b$ for some $a, b \in E_P$. Then, since $\causes_P$ is a partial order, either:
		\begin{itemize}
			\item $a = b$, in which case also $a \causes_{\lang{P}} b$ since $\causes_{\lang{P}}$ is reflexive; or

			\item $a \causes_P b$ but $a \neq b$, in which case $a$ will precede $b$ in every (maximal) trace of $P$. Furthermore, $\lang{P}$ will not be empty, and then $a \causes_{\lang{P}} b$ by its first rule.
		\end{itemize}

		\item Suppose that $a \causes_{\lang{P}} b$ for some $a, b \in E_P$. Then, by definition of $\causes_{\lang{P}}$, either:
		\begin{itemize}
			\item $a = b$, in which case also $a \causes_P b$ since $\causes_P$ is reflexive; or

			\item (1) $\exists x, y, z \in E_P^*: xaybz \in \lang{P}$ and (2) $\forall \hat{x}, \hat{y}, \hat{z} \in E_P^*: \hat{x}b\hat{y}a\hat{z} \notin \lang{P}$.
			Suppose, for the sake of contradiction, that $a \not \causes_P b$. Then either $a >_P b$, which contradicts (1), or $a \concurrent_P b$, which contradicts (2) since there should then also exist some trace in $\lang{P}$ in which $b$ precedes $a$.
			We can thus conclude that $a \causes_P b$.
		\end{itemize}
		We note that the transitive rule for $\causes_{\lang{P}}$ is subsumed by the other two and does not need to be considered separately. Nevertheless, a straightforward inductive argument suffices to cover this.
	\end{itemize}
\end{proof}

% \begin{lemma}
% 	Let $R_t, R_1, \ldots, R_n$ be pomsets such that $\lang{\shuffle_{R_t}(R_1, \ldots, R_n)}$ is not empty and equals $\lang{R}$ for some pomset $R$.
% 	Let $a$ and $b$ be events of $R_i$ and $c$ an event of $R_j$ ($i \neq j$).
% 	If $a \concurrent b$ in $R_i$, then either $a, b < c$, or $c < a, b$, or $a, b \concurrent c$ in $R$.
% \end{lemma}
\lemrelation*
\begin{proof}
	If no subscript is given, in this proof, $<$ and $\concurrent$ refer to $<_P$ and $\concurrent_P$.

	We start by making three observations:
	\begin{enumerate}[label=(\arabic*)]
		\item As noted in \cref{sec:posets}, if $a \concurrent_{P_i} b$, then there must exist traces in $\lang{P_i}$ such that $a$ occurs before $b$ in one, and after $b$ in another.
		Furthermore, since all traces of a poset can be obtained from an arbitrary one by swapping adjacent concurrent events, there must also exist traces such that $a$ and $b$ are adjacent to each other, in both orders, e.g., $vabw$ and $vbaw$ for some $v, w$.

		\item The causal relation between two events in $P_i$ remains the same in $P$.
		The shuffle on trajectories operator does not change the internal order of the traces of its operands, meaning that, for example if $d$ occurs before $e$ in every trace of $P_i$, then it will also occur before $e$ in every trace of $P$. Similarly, if $P_i$ both contains traces in which $d$ occurs before $e$ and traces in which $d$ occurs after $e$, then so will $P$.
		In particular, this means that, since $a \concurrent_{P_i} b$, also $a \concurrent b$.

		\item If two events are concurrent in $P_i$ and can thus be swapped in traces of $P_i$ when adjacent, then they can also be swapped in traces of $P$ when there are no other events from $P_i$ in between them. The shuffle on trajectories operator does not make any distinction between the corresponding traces from $P_i$.
		Additionally, we note that this implies that the two events are concurrent with all events (from other operands) in between.
	\end{enumerate}

	% Now consider $\lang{R}$. Since $a$ and $b$ can be freely swapped when adjacent, it follows that, for all $u, v, w$, where $v$ contains no events from $R_i$, $uavbw \in \lang{R} \Leftrightarrow ubvaw \in \lang{R}$.

	Consider the relation between $a$, $b$ and $c$ in $P$. As observed in (2), $a \concurrent b$.
	\begin{itemize}
		\item Suppose that $a < c$ and $c < b$. It then follows by transitivity (recall that $\leq$ is a partial order) that $a < b$, which contradicts $a \concurrent b$.
		Analogously, we can exclude the case where $b < c$ and $c < a$.

		\item Suppose that $a < c$ and $b \concurrent c$. Then, there must exist a trace $uavcwbx$ in $\lang{P}$ for some $u, v, w, x$, i.e., a trace where $a$ occurs before $c$ and where $b$ occurs after $c$.
		If $vw$ contains no events from $P_i$, then it follows from $a \concurrent b$ and (3) that $\lang{P}$ must also contain $ubvcwax$, which contradicts $a < c$.
		Otherwise, we systematically remove the events from $P_i$ in $vw$:
		\begin{itemize}
			\item Let $d$ be the leftmost event from $P_i$ in $vw$ such that $d < b$. Then, since we must be able to swap adjacent concurrent events in traces of $P_i$ until $a$ and $b$ are adjacent (as noted in (1)), $d$ must be concurrent with all events from $P_i$ to the left of it in $avw$ (including $a$). We can thus swap $d$ with its left neighbour from $P_i$ until it has been swapped with $a$, thus reducing the number of events from $P_i$ in $vw$ by 1.
		 	We repeat this until there are no events left in the remainder of $vw$ that fit this description.

			\item Analogously, let $e$ be the rightmost event from $P_i$ in the remainder of $vw$ such that $a < e$. Then $e$ must be concurrent with all events from $P_i$ to the right of it until at least $b$, so we swap it with its right neighbour from $RPi$ until it has been swapped with $b$. Again, we repeat this until there are no events left that fit this description.

			\item All events from $P_i$ remaining in $vw$ must now be concurrent with both $a$ and $b$. We can thus keep swapping $a$ with its right neighbour (or, alternatively, $b$ with its left neighbour) from $P_i$ until there are no events from $P_i$ remaining between $a$ and $b$.
		\end{itemize}

		If, at some point, $a$ has passed $c$ and is now to the right of it, then this contradicts $a < c$. If $a$ remains to the left of $c$ and $b$ to its right, then we can swap $a$ and $b$ to contradict $a < c$. Otherwise, if $b$ has passed $c$ and is now to the left of it, then $b$ must be concurrent with all events to the right of it at least until $c$ (as noted in (3)). We can then swap it with its right neighbour (from $P$) until it has been swapped with $c$, after which we can swap it with $a$ to contradict $a < c$.

		As all cases lead to a contradiction, we can thus conclude that it is not possible that $a < c$ and $b \concurrent c$. Analogously, we can also exclude the cases where $c < a$ and $b \concurrent c$, where $b < c$ and $a \concurrent c$, and where $c < b$ and $a \concurrent c$.
	\end{itemize}

	As we have excluded all other cases, this proves that either $a, b < c$, or $c < a, b$, or $a, b \concurrent c$.
	\qed
\end{proof}

\else
\fi

\end{document}